\documentclass[letterpaper, 10 pt, conference]{ieeeconf}  

\IEEEoverridecommandlockouts                              

\usepackage{amsmath} 
\usepackage{amssymb}  
\usepackage{amsfonts}
\usepackage{mathtools} 
\usepackage{fourier}
\usepackage{xcolor}  
\usepackage{tikz}
\usetikzlibrary{calc}

\usepackage{amsthm}

\newtheorem{thm}{Theorem}
\newtheorem{defn}{Definition}

\newtheorem{lem}{Lemma}

\newtheorem{asm}{Assumption}

\newtheorem{remark}{Remark}

\usepackage{ifthen}
\newcommand\suppress[1]{} 
\newcommand\switchVersion[3]{\ifthenelse{\boolean{#1}}%
{#2\suppress{#3}}%
{\suppress{#2}{\color{magenta}#3}}}
\newboolean{shortVersion}
\setboolean{shortVersion}{true} 
\newcommand{\squeezeupSmall}{\vspace{-2mm}}
\newcommand{\squeezeupMid}{\vspace{-4mm}}

\usepackage{lipsum}
\usepackage{etoolbox}
\newcommand{\tinydisplayskip}{%
  \setlength{\abovedisplayskip}{4pt}%
  \setlength{\belowdisplayskip}{4pt}%
  \setlength{\abovedisplayshortskip}{0pt}%
  \setlength{\belowdisplayshortskip}{0pt}} 

\renewcommand{\phi}{\varphi}
\renewcommand{\epsilon}{\varepsilon}

\newcommand{\maji}[1]{\ensuremath{\mathbb{#1}}}
\newcommand{\bigfun}[1]{\ensuremath{\mathcal{#1}}}
\newcommand{\UUU}{{\mathcal{U}}}

\newcommand{\R}{{\maji{R}}}
\newcommand{\Rnonneg}{{\maji{R}}_{\geq 0}}
\newcommand{\Rpos}{{\maji{R}}_{> 0}}
\newcommand{\Z}{{\maji{Z}}} 
\newcommand{\Zpos}{{\maji{Z}}_{> 0}}
\newcommand{\N}{\maji{Z}_{\geq 0}}
\newcommand{\C}{C}
\newcommand{\CCC}{\mathcal{C}}
\newcommand{\SSS}{{\bigfun{S}}}

\newcommand{\GGG}{\bigfun{G}}
\newcommand{\VVV}{\mathcal{V}}

\newcommand{\fw}{\rightarrow}
\newcommand{\bw}{\leftarrow}

\newcommand{\xto}[1]{\xrightarrow{#1}}

\newcommand{\powerset}{{\ensuremath{\wp}}}
\newcommand{\ens}[1]{\left\{ #1 \right\}}

\newcommand{\Ball}[2]{\mathcal{B}_{#2}(#1)} 

\newcommand{\len}{{{\text{len}}}}

\newcommand{\model}{{\mathcal{M}}}

\newcommand{\traj}{\sigma}
\newcommand{\Traj}{\mathit{Traj}}
\newcommand{\FTraj}{\mathit{FTraj}}
\newcommand{\IRun}{\mathit{Run}}
\newcommand{\FRun}{\mathit{FRun}}

\newcommand{\FPlay}{\mathit{FPlay}}
\newcommand{\FPlayi}{\FPlay_1}
\newcommand{\FPlayii}{\FPlay_2}

\newcommand{\control}[2]{{#1}/{#2}}
\newcommand{\ini}{\text{in}}

\DeclareMathOperator{\logic}{2-LTL}
\newcommand{\Coloneqq}{\mathrel{\mathop{::}=}}

\newcommand{\Dis}[2]{[{#1}]_{#2}} 

\DeclareMathOperator{\Inf}{Inf}
\DeclareMathOperator{\MP}{MP}

\newcommand{\mppg}{mean-payoff parity game}
\newcommand{\mppgs}{\mppg{}s}

\newcommand{\pmay}{{\rho_\exists}}
\newcommand{\pmust}{{\rho_\forall}}

\title{\LARGE \bf
Symbolic Self-triggered Control of Continuous-time Non-deterministic Systems without Stability Assumptions for 2-LTL Specifications
}

\author{Sasinee Pruekprasert, Clovis Eberhart, and J\'{e}r\'{e}my Dubut
\thanks{The authors are supported by ERATO HASUO Metamathematics for Systems Design Project (No. JPMJER1603), JST. J. Dubut is also supported by Grant-in-aid No. 19K20215, JSPS.}
\thanks{The authors are with National Institute of Informatics, Hitotsubashi 2-1-2, Tokyo 101-8430, Japan
        {\small \{sasinee, eberhart, dubut\}@nii.ac.jp}.}%
\thanks{C. Eberhart and J. Dubut are also affiliated with the Japanese-French Laboratory for Informatics.}%
}

\begin{document}
{\onecolumn\large 
\noindent\textcopyright\ 2020 IEEE.  Personal use of this material is permitted.  Permission from IEEE must be obtained for all other uses, in any current or future media, including reprinting/republishing this material for advertising or promotional purposes, creating new collective works, for resale or redistribution to servers or lists, or reuse of any copyrighted component of this work in other works.}
\newpage
\twocolumn

\maketitle
\thispagestyle{empty}
\pagestyle{empty}

\begin{abstract}
We propose a symbolic self-triggered controller synthesis procedure for non-deterministic
continuous-time nonlinear systems without stability assumptions.
The goal is to compute a controller that
satisfies two objectives.
The first objective is represented as a specification in a fragment of LTL, which we call 2-LTL.
The second one is an energy objective, in the sense that control inputs are
issued only when necessary, which saves energy.
To this end, we first quantise the state and input spaces, and then translate
the controller synthesis problem to the computation of a winning strategy in a
mean-payoff parity game.
We illustrate the feasibility of our method on the example of a navigating
nonholonomic robot.
\end{abstract}

\section{Introduction}
Not only has self-triggered control been a hot academic research topic in recent years, but it
also provides a variety of practical implementations~\cite{HJT2012}. By performing sensing 
and actuation only when needed, self-triggered control is well-known as an energy-aware 
control paradigm to save communication resources for Networked Control 
Systems~\cite{HAD2018}. The lifespan of battery-powered
devices can be prolonged by reducing their energy consumption~\cite{HJT2012} and 
the communication load of nonholonomic robots significantly reduced in comparison to 
using periodic controllers~\cite{SEMGL2019}. 
However, previous research on self-triggered control of continuous-time systems only studies
simple specifications such as stability \cite{anta2010} and reach-avoid or safety problems \cite{Hashimoto2019,  hashimoto20}. 


The main reason for this limitation is that those approaches are based on reachability 
analysis. The main novelty of our work is to use techniques from game theory to go beyond 
reach-avoid and safety specifications for self-triggered control.
Game theory, and in particular parity games \cite{emerson91}, is a well-known technique
to deal with expressive logic like the $\mu$-calculus \cite{emerson91} and
$\text{CTL}^*$ \cite{friedmann13}, as the parity winning conditions provide
complex scenarios and strategies while keeping computability.
In particular, parity games can be
used for control synthesis of reactive systems under Linear Temporal Logic (LTL) specifications \cite{luttenberger19}. 
On the other hand, quantitative games 
such as mean-payoff games~\cite{ehrenfeucht79}
have been adapted to quantitative control specifications \cite{Pru2016, Ji2018}.
One of such specifications is the mean-payoff threshold problem
for the average control-signal length of self-triggered controllers. 
This threshold provides guarantees for the energy-saving and
communication-reduction performance of the controller:
the greater the average length of a signal is, the less often
the controller needs to perform sensing and actuation, and fewer
commands are sent across the network. 
In this paper, 
we deal with this threshold problem together with a logical specification
using mean-payoff parity games, which combine mean-payoff games and parity games.
The logical specifications can be dealt with the parity side, while the 
average 
signal length threshold
can be seen as a threshold problem in a mean-payoff game.

Our procedure is based on the symbolic control approach, which synthesises correct-by-design controllers of continuous-state systems.
In this approach, we first construct a symbolic model, which is a discrete abstraction of the continuous-state system, based on approximate simulation or bisimulation.
Then, we synthesise a symbolic controller
and leverage its control strategy to control the continuous-state system.
This technique allows us to synthesise provably-correct controllers for complex specification such as LTL specifications, which can hardly be enforced with conventional control methods.
However, previous symbolic control algorithms for continuous-time nonlinear systems under LTL specifications need stability assumptions (e.g.,~\cite{Kido2018, Zamani2014}), which do not hold in many systems.  
Symbolic control without stability assumptions is enforced on simpler classes of specifications such as reach-avoid~\cite{Hashimoto2019, MR2019, Zamani2012}. 
The work closest to ours is~\cite{Hashimoto2019}, in which the authors propose
an algorithm to synthesise symbolic self-triggered controllers for discrete-time deterministic
systems under reach-avoid specifications.

This work proposes a symbolic self-triggered control procedure for
continuous-time non-deterministic nonlinear systems without stability
assumptions for specifications represented by a fragment of LTL, which
we call 2-LTL.
To the best of our knowledge, 
our work is the first to study symbolic control of continuous-time nonlinear systems without stability assumptions for a class of LTL specifications that is strictly more expressive than reach-avoid.
Our procedure operates in several steps: (1) constructing a finite 
symbolic model of the continuous system, 
(2) 
translating the self-triggered controller 
synthesis problem 
on the symbolic model 
into a mean-payoff parity game problem, 
(3) constructing a winning strategy for the mean-payoff parity game, 
(4) translating the strategy back into a controller for the symbolic model, and
(5) translating the controller for the symbolic model back into one for the continuous system.

\switchVersion{shortVersion}{}{
The paper is organised as follows. In Section~\ref{section: framework}, we introduce the 
systems of interest,
together with their interpretation as state-transition models, and their controllers. We then 
formulate our problem of finding a controller for those systems which satisfies some logical 
specification in a fragment of LTL (called \emph{$\logic$}) and some self-trigger conditions 
in Section~\ref{section: problem}. We 
then translate this problem into a control synthesis problem for a symbolic model in 
Section~\ref{section: reduction}. In Section~\ref{sec: algo}, we provide some basics 
on mean-payoff parity games, and a translation of our 
problem into synthesising a strategy for such a game. We then illustrate our method on the
example of a nonholonomic robot in Section~\ref{section: illus example}. Finally, 
Section~\ref{section: conclusion} concludes the paper.
}

\textit{Notation:} We denote vectors in $\mathbb{R}^m$ by $x = \begin{bmatrix} x_1 &
\cdots & x_m\end{bmatrix}^{\intercal}$. For such a vector $x$, we use $\lVert x \rVert$ for its infinity norm $\max\,\{\,\lvert x_i\rvert ~ \mid\, i \in \{1, \ldots, m\}\}$.
Given $x \in \mathbb{R}^m$ and $r \in \mathbb{R}_{> 0}$, we write $\Ball{x}{r}$
for the ball $\{y \in \mathbb{R}^m \,\mid\, \lVert x-y \rVert \leq r\}$ of centre $x$ and radius $r$. 
Finally, given a set $X$, we denote its powerset $\{Y \,\mid\, Y\subseteq X\}$ by $\powerset(X)$.

\section{Control Framework} \label{section: framework} 

\subsection{System} \label{subsection: system} 
We formalise a non-deterministic continuous-time nonlinear system as a 6-tuple
$
  \Sigma = (X, X_\ini, U, {\UUU}, \xi^\fw, \xi^\bw),
$
where
$X \subseteq \R^n$ is a bounded convex state space,
$X_\ini \subseteq X$ is a space of initial states,
$U \subseteq {\R}^m$ is a bounded convex space of control inputs,
${\UUU}$ is a set of control signals of the form
$[0,T] \to U$ that assign a control input at each time in the interval $[0,T]$ with $T \in \Rpos$,
and $\xi^\fw$, $\xi^\bw$ $: \R^n \times \UUU \times \Rnonneg \to \powerset(\R^n)$ are functions
such that $\xi_{x,u}^\fw(0) = \xi_{x,u}^\bw(0) = \{x\}$.
Intuitively, given a state $x \in \R^n$, a signal $u \in \UUU$, and a time $t \geq 0$,
$\xi_{x,u}^\fw(t)$ (\emph{resp.} $\xi_{x,u}^\bw(t)$) is the reachable states
from $x$ (\emph{resp.} the set of states from which the system can reach $x$) under the control signal $u$ at time $t$.

The system is defined on the whole Euclidean space $\R^n$, but we are
only interested in its behaviour on a bounded subspace $X$ because the
quantities involved in systems are physically bounded, as observed
in~\cite{Zamani2012}.
For technical reasons, we also assume that the distance 
from $X_\ini$ to the boundary of $X$ is positive.

The system $\Sigma$ is said to be \emph{forward and backward complete} \cite{Angeli1999} if
$\xi_{x,u}^\fw(t) \neq \emptyset$ and $\xi_{x,u}^\bw(t) \neq \emptyset$ for all
$(x,u) \in \R^n \times \UUU$ and all $t \in [0, \len(u)]$, where $\len(u)= T$ is the length of the signal $u:[0,T] \to U$.
In other words, under any control signal,
there exist a state reachable from $x$ and a state that reaches $x$
at any time within the signal length.

\begin{defn}\label{defn: inc fw and bw complete} 
A system $\Sigma$ is \emph{incrementally forward and backward complete} if
it is forward and backward complete, and, for each $u \in \UUU$,
there exist functions $\beta^\fw_u, \beta^\bw_u: \Rnonneg \times \Rnonneg \to \Rnonneg$
such that 1) for any $t \in \Rnonneg$, $\beta^\fw_u(\_, t)$ and
$\beta^\bw_u(\_, t)$ are strictly increasing and their limits at
$+\infty$ is $+\infty$,
and 2) for any $x_1,x_2 \in \R^n$, $u \in {\UUU}$, and $t \leq \len(u)$,
		\switchVersion{shortVersion}{
		\begin{enumerate}
		\item[2.1)] $\forall (x_1',x_2') \in \xi^\fw_{x_1,u}(t) \times \xi^\fw_{x_2,u}(t),
			\lVert x_1' - x_2'\rVert \leq 
      		\beta^\fw_u(\lVert x_1 - x_2 \rVert,t)\rlap{,}  
      		$
		\item[2.2)] $\forall (x_1',x_2') \in \xi^\bw_{x_1,u}(t) \times \xi^\bw_{x_2,u}(t),
			\lVert x_1' - x_2'\rVert \leq 
        	\beta^\bw_u(\lVert x_1 - x_2 \rVert,t)\rlap{.}
			$ 
		\end{enumerate}		
		}{
		\begin{enumerate}
		\item[2.1)] for all $(x_1',x_2') \in \xi^\fw_{x_1,u}(t) \times \xi^\fw_{x_2,u}(t)$,
			\begin{equation} \label{eq: inc fw complete}
			\lVert x_1' - x_2'\rVert \leq 
      \beta^\fw_u(\lVert x_1 - x_2 \rVert,t)\rlap{,}
			\end{equation} 
		\item[2.2)] for all $(x_1',x_2') \in \xi^\bw_{x_1,u}(t) \times \xi^\bw_{x_2,u}(t)$,
			\begin{equation} \label{eq: inc bw complete}
			\lVert x_1' - x_2'\rVert \leq 
        \beta^\bw_u(\lVert x_1 - x_2 \rVert,t)\rlap{.}
			\end{equation} 
		\end{enumerate}
		}
\end{defn}

Notice that incremental forward and backward completeness does not depend on 
 the state space $X$, but on $\R^n$. There may exists $(x,u)
\in X \times \UUU$ such that $\xi^\fw_{x,u}(t) \cap X = \emptyset$,
i.e., the system runs out of the desired state space.

\begin{asm}\label{asm: inc fw and bw complete} 
The system $\Sigma$ is incrementally forward and backward complete. 
\end{asm}

Assumption~\ref{asm: inc fw and bw complete} is similar to the one used
in~\cite{Zamani2012}, but adapted to non-deterministic systems and
taking backward dynamics into account. 
The intuition behind Assumption~\ref{asm: inc fw and bw complete} is
that the distance between the states reached from two starting points
can be bound by an expression that depends only on the distance
between those starting points, the control signal, and the run time.
In addition, we require the following assumption.

\begin{asm}\label{asm: Lipchitz} 
For any control signal $u \in \UUU$,
we have functions $\alpha^\fw_u, \alpha^\bw_u: \Rnonneg \times [0, \len(u)] \to \Rnonneg$
such that 1) for any $t \in \Rnonneg$, $\alpha^\fw_u(\_,t)$ and
$\alpha^\bw_u(\_,t)$ are increasing, and 2) for any $x_1, x_2 \in X$ and any $t \in [0, \len(u)]$, we have 
\begin{enumerate}
\item[2.1)] for any $y_2 \in \xi_{x_2,u}^\fw(t),
	\rVert x_1 - y_2 \lVert \leq \alpha^\fw_u(\rVert x_1 - x_2 \lVert, t)\rlap{,}
	$
\item[2.2)] for any $y_2 \in \xi_{x_2,u}^\bw(t),
    \rVert x_1 - y_2 \lVert \leq \alpha^\bw_u(\rVert x_1 - x_2 \lVert, t)\rlap{.}
	$
\end{enumerate}
\end{asm}

Assumption~\ref{asm: Lipchitz} basically states that the set of
reachable states cannot be arbitrarily far from the starting state
(for a given input signal and run time).
Those functions $\beta^\fw_u$, $ \beta^\bw_u$, $\alpha^\fw_u$, and $\alpha^\bw_u$ can typically be computed using Lyapunov functions (see \cite{Zamani2012, Angeli1999, Angeli2002} for details).

\subsection{State-transition Model} \label{subsection: model}
Let us first introduce general definitions for state-transition models
and their controllers.
For this paper, a state-transition model is given by a quadruple
$\model = (Y, Y_\ini, \VVV, \rightarrow),$
where $Y$ is either a continuous or a discrete state space, $Y_\ini \subseteq Y$ is a set of initial states, $\VVV$ is a set of control signals, and 
$ \rightarrow\, \subseteq Y \times \VVV \times Y$ is the transition relation.  
For a given state $y \in Y$,
a sequence $y_0 u_0 y_1 u_1 \ldots \in Y (\VVV Y)^\omega$ (\emph{resp.} $y_0 u_0 y_1 \ldots u_{l-1} y_l \in Y (\VVV Y)^*$) is a \emph{run} (\emph{resp.} a \emph{finite run})
generated by $\model$ \emph{starting from} the state $y$
if $y_0 = y $ and
 $(y_i, u_i, y_{i+1}) \in \to$ for any $i \in \N$ (\emph{resp.} $i \in \{0,\ldots,l-1\}$).
Let $\IRun(\model, y)$ (\emph{resp.} $\FRun(\model,y)$) denote the set of all runs (\emph{resp.} finite runs) generated by $\model$ from $y$.
Let 
$\IRun(\model)= \bigcup_{y\in Y_\ini} \IRun(\model, y)$ and $\FRun(\model)= \bigcup_{y\in Y_\ini} \FRun(\model, y)$.

We define the state-transition model of $\Sigma$ as follows. 

\begin{defn}\label{defn: model}
The state-transition model $\model(\Sigma)$ of a system $\Sigma= (X, X_\ini, U, {\UUU}, \xi^\fw, \xi^\bw)$ is 
$
\model(\Sigma) = (X, X_\ini, \UUU,\Delta),
$
where the transition relation $\Delta \subseteq X \times \UUU \times X$
is given by 
\switchVersion{shortVersion}{\squeezeupSmall}{}
\begin{align*} \label{eq: trans model}
  (x, u, x') \in \Delta \text{~iff~}
    & x' \in \xi_{x,u}^\fw(\len(u)) \text{,~}  x \in
      \xi_{x',u}^\bw(\len(u)) \\
    & \text{~and, for all $t \leq \len(u)$, $\xi^\fw_{x,u}(t)
      \subseteq X$.}
\end{align*}
\end{defn}
\switchVersion{shortVersion}{\squeezeupSmall}{
\begin{remark}$\model(\Sigma)$ is called a \emph{symbolic model} in~\cite{Zamani2012}.
In this paper, however, we reserve the term \emph{symbolic model} for the discrete-state system in Section  \ref{subsection: symbolic model}.
\end{remark}
}

Runs of a model are discrete sequences of states, but the system runs
in continuous time.
In order to fill this gap, we introduce the notion of trajectory to
match these discrete runs to continuous sequences of states.
\begin{defn}\label{defn: trajectory}
  A \emph{trajectory} of a system $\Sigma$ starting from a state $x
  \in X$ induced by a run $x_0 u_0 x_1 u_1 x_2 \ldots \in
  \IRun(\model(\Sigma), x)$ is a function $\traj: \Rnonneg \to X$ such
  that, for all $k \in \N$ and all $t \in \big[\sum_{i < k} \len(u_i), \sum_{i \leq k} \len(u_i)\big]$,
  \switchVersion{shortVersion}{\squeezeupSmall
  \begin{align*}
  \traj&(t) \in \xi_{x_k,u_k}^\fw\big( t - \sum_{i < k} \len(u_i)\big) 
   \cap \xi_{x_{k+1},u_k'}^\bw\big(\sum_{i \leq k} \len(u_i) - t\big),\\
  &\text{where }u_k'(s) = u_k\big(s+t- \displaystyle\sum_{i < k} \len(u_i)\big),  
   \forall s \in [0, \displaystyle\sum_{i \leq k} \len(u_i)-t]. 
  \end{align*} 
  }{
  \begin{equation*}  
   \traj(t) \in \xi_{x_k,u_k}^\fw\big( t - \sum_{i < k} \len(u_i)\big) 
   \cap \xi_{x_{k+1},u_k'}^\bw\big(\sum_{i \leq k} \len(u_i) - t\big),  
  \end{equation*}
  where 
  \begin{equation*} 
  u_k'(s) = u_k\big(s+t- \displaystyle\sum_{i < k} \len(u_i)\big)  
  \text{ for all }
   s \in [0, \displaystyle\sum_{i \leq k} \len(u_i)-t].
  \end{equation*}
  } 
\end{defn}
Let $\Traj(\Sigma, x, r)$ be the set of trajectories of $\Sigma$ that
are induced by a run $r \in \IRun(\model(\Sigma), x)$.
For any finite run $r_f \in \FRun(\model(\Sigma), x)$,
$\FTraj(\Sigma, x, r_f)$ is the set of finite trajectories defined in
the same way.
Let $\Traj(\Sigma, x) = \bigcup_{r \in \IRun(\model(\Sigma), x)}
\Traj(\Sigma, x, r)$, and $\Traj(\Sigma) = \bigcup_{x \in X_\ini}
\Traj(\Sigma,x)$.

\subsection{Controlled System} \label{subsection: controlled system}

In this section, we define controllers and controlled systems, and
explain the self-triggered control process.
First, we define \emph{model controllers} of $\model = (Y, Y_\ini,
\VVV, \rightarrow)$.

\begin{defn}\label{defn: model controller}
A \emph{model controller} of a state-transition model $\model$ is a function $\CCC: \FRun(\model) \to \VVV$.  
\end{defn}
 
Let $\control{\CCC}{\model}$ denote the state-transition model $\model$ controlled under $\CCC$. 
A run $y_0 u_0 y_1  \ldots \in \FRun(\model,y)$ (\emph{resp.} a finite
run $y_0 u_0 y_1   \ldots u_{l-1} y_l  \in \IRun(\model,y)$) is
generated by $\control{\CCC}{\model}$ if it satisfies the following
conditions: 1) $y_0 = y$, and 2) for all $i \in \N$ (\emph{resp.} for
all $i \in [0, l-1]$), we have $u_{i} = \CCC(y_0 u_0 \ldots y_i)$.
Then, let $\IRun(\control{\CCC}{\model}, y)$ (\emph{resp.} $\FRun(\control{\CCC}{\model}, y)$) denotes the set of all runs (\emph{resp.} finite runs) generated by $\control{\CCC}{\model}$ from $y \in Y$. 
Let $\IRun(\control{\CCC}{\model})= \bigcup_{y\in Y_\ini} \IRun(\control{\CCC}{\model}, y)$ and $\FRun(\control{\CCC}{\model})= \bigcup_{y\in Y_\ini} \FRun(\control{\CCC}{\model}, y)$.

\begin{defn}\label{defn: controller}
A \emph{controller} of $\Sigma = (X, X_\ini, U, {\UUU}, \xi^\fw, \xi^\bw)$ is a function $\C: \FRun(\model(\Sigma)) \to \UUU$.  
\end{defn}

Notice that a controller $\C$ of a system $\Sigma$ is defined based on its state-transition model $\model(\Sigma)$. This is because the controller issues control signals based on runs, which only track the states at the end of each signal.
Since Definition~\ref{defn: controller} is coherent with Definition~\ref{defn: model controller}, we can also regard $\C$ as a model controller of $\model(\Sigma)$.
Hence, we also use $\IRun(\control{\C}{\model(\Sigma)})$
(\emph{resp.} $\FRun(\control{\C}{\model(\Sigma)})$) to denote the sets of
runs (\emph{resp.} finite runs) of $\control{C}{\model(\Sigma)}$ from
initial states.
Furthermore, we use $\control{\C}{\Sigma}$ to denote the \emph{system}
$\Sigma$ \emph{controlled under} the controller $\C$, and define the
trajectories of $\control{\C}{\Sigma}$ in the same way as in
Definition~\ref{defn: trajectory}. 
Thereby, the definitions of
trajectories $\Traj(\control{C}{\Sigma})$ and
$\FTraj(\control{C}{\Sigma})$ carry over to controlled systems
directly.

\begin{figure}[t]
      \centering
      \includegraphics[scale=0.25]{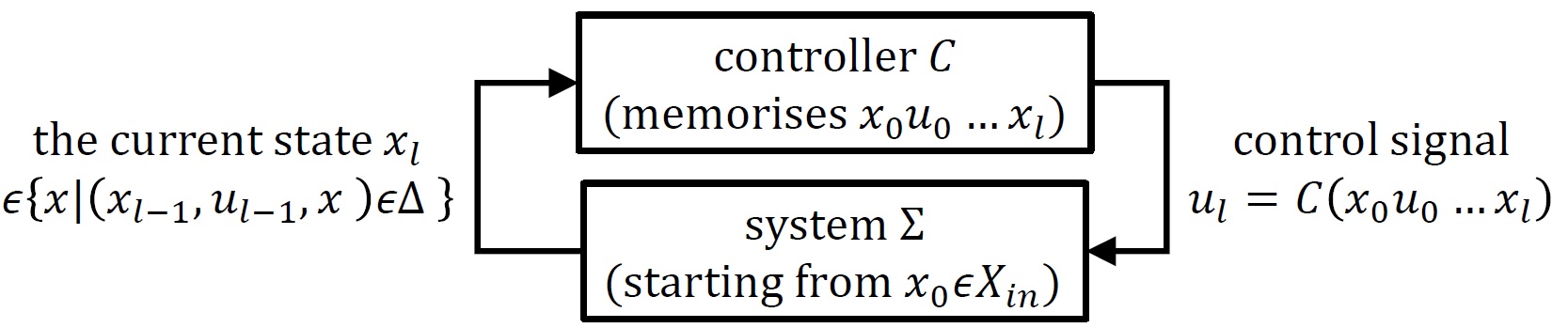}
      \switchVersion{shortVersion}{\squeezeupSmall}{}
      \caption{Overview of the self-triggered control process. Based on previously observed states and issued control signals, the controller issues a control signals to control the system.}
      \label{fig: control process}
\switchVersion{shortVersion}{\squeezeupMid}{}
\end{figure}
The overview of the control process is illustrated in Fig~\ref{fig: control process}.
First, the controller observes the initial state $x_0 \in X_\ini$ and 
issues the control signal $u_0 = \C(x_0)$. 
Then, to preserve energy, the controller is inactive throughout the duration of the control signal $u_0$. Namely, the longer the signal length $\len(u_0)$ is, the more energy is preserved.
Since the system is non-deterministic, there are several states that can possibly be reached under the signal $u_0$.
After the signal ends (at time $\len(u_0)$), 
the controller becomes active and resolves the non-determinism by detecting the actual current state $x_1$ 
and issue a new control signal $u_1 = \C(x_0 u_0 x_1)$. 
The process is then repeated.

\section{Problem Formulation} \label{section: problem} 

Our goal is to synthesise a controller that satisfies two control objectives.
The first objective is described as a $\logic$ formula. 
The second one is an energy-preservation objective: to ensure that the average length of the issued control signals is above a given threshold.

\subsection{$\logic$ Specification} \label{subsection: logic} 

We model the first objective using a fragment of LTL, which we call $\logic$.
Let $AP$ denote the set of atomic propositions, i.e., assertions that
can be either true or false at each state $x \in X$.
Let $P: X \to \powerset(AP)$ assign the set of atomic propositions
that hold at each state.

\begin{defn}
  Let $\logic$ be the logic whose formulas are the $\Phi$'s generated by the following grammar:
  {\switchVersion{shortVersion}{\tinydisplayskip}{}
  \begin{align*}
    \phi & \Coloneqq \top \mid p \mid  \neg \phi \mid \phi \vee \phi \\
    \Phi & \Coloneqq 
      \Diamond \phi \mid \Box \phi \mid \Box\Diamond \phi \mid \Diamond\Box \phi
      \mid \Phi \vee \Phi \mid \Phi \wedge \Phi \rlap{,}
  \end{align*} 
  }
  where $p\in AP$ is an atomic proposition.
\end{defn}

We call $\phi$'s and $\Phi$'s state formulas and path formulas,
respectively. 
A logic specification is written as a path formula.
Here, $\Box$ and $\Diamond$ have the usual interpretation of LTL.
%
A state $x\in X$ satisfying (\emph{resp.} not satisfying) a state formula $\phi$ is denoted by $x \vDash \phi$ (\emph{resp.} $x \not\vDash \phi$). 
We also use the same notations $\traj \vDash \Phi$ and $\traj   \not\vDash \Phi$ for a trajectory $\traj: \Rnonneg \to X$ and a path formula $\Phi$. For every state $x \in X$, 
$x \vDash \phi$ is defined as follows: 
{\switchVersion{shortVersion}{\tinydisplayskip}{}
\begin{align*}
	x &\vDash \top &  x &\vDash  p \text{ if } p \in P(x)\\
	x &\vDash \neg\phi \text{ if } x \not\vDash \phi
      &  x &\vDash \phi_1 \vee \phi_2 \text{ if } x \vDash \phi_1 \text{ or }x \vDash \phi_2\rlap{,}
\end{align*}
}
and for all $\traj: \Rnonneg \to X$, $\sigma \vDash \Phi$ is defined as follows:
{\switchVersion{shortVersion}{\tinydisplayskip}{}
\begin{align*}
	\traj &\vDash  \Diamond\phi  \text{~~if~~} \exists t \in \Rnonneg,\, \traj(t) \vDash \phi\\
	\traj &\vDash  \Box\phi  \text{~~if~~} \forall t \in \Rnonneg,\, \traj(t) \vDash \phi\\
	\traj &\vDash  \Box\Diamond\phi  \text{~~if~~} \forall t \in \Rnonneg,\, \exists t' > t,\, \traj(t') \vDash \phi\\
	\traj &\vDash  \Diamond\Box\phi  \text{~~if~~}\exists t \in \Rnonneg,\, \forall t' > t,\, \traj(t') \vDash \phi\\ 
	\traj &\vDash \Phi_1 \vee \Phi_2  \text{~~if~~}\traj \vDash \Phi_1 \text{ or } \traj \vDash \Phi_2\\
  \traj &\vDash \Phi_1 \wedge \Phi_2  \text{~~if~~} \traj\vDash \Phi_1 \text{ and } \traj \vDash \Phi_2\rlap{.}
\end{align*}
}
One objective of a controller $\C$ is to control the system in such a way that all  trajectories in $ \Traj(\control{\C}{\Sigma})$ satisfy a given $\logic$ path formula $\Phi$. 
Notice that the class of $\logic$ specifications is more general than the reach-avoid specifications, 
which is studied in \cite{Hashimoto2019,MR2019,hashimoto20}.
For example, 
we can express the logic specification to reach $\mathtt{target\_region}$ while avoiding $\mathtt{unsafe\_region}$ using the $\logic$ formula $\Diamond \mathtt{target\_region} \wedge \Box \neg\mathtt{unsafe\_region}$.



\subsection{Controller Synthesis Problem} \label{subsection: problem continuous} 

\begin{defn} \label{defn: problem continuous}
Given a system $\Sigma = (X, X_\ini, U, {\UUU}, \xi^\fw, \xi^\bw)$, a set $AP$ of atomic propositions, a function 
$P: X \to  \powerset(AP)$,
a $\logic$ formula $\Phi$, 
and a threshold $\nu \in \Rpos$,
the \emph{controller synthesis problem} is 
to synthesise  a controller $\C:\FRun(\model(\Sigma)) \to \UUU$ 
such that 
\begin{itemize}
  \item all finite runs in $\FRun(\control{C}{\model(\Sigma)})$ can be
    extended to an infinite run in $\IRun(\model(\Sigma))$,
\item $\traj \vDash \Phi$ for any $\traj \in \Traj(\control{\C}{\Sigma})$, and
\item 
$\displaystyle
\lim_{h \to \infty} \frac{1}{h} \sum_{i = 1}^ h \len(u_i) > \nu
$ for any $x_0 u_0 \ldots \in \IRun(\control{\C}{\model(\Sigma)})$,
\end{itemize}
or determine that such a controller $\C$ does not exist.
\end{defn}
The first condition simply ensures that the controlled system does not
reach a deadlock, while the other two conditions are the actual
control objectives.

\section{Problem Reduction to Symbolic Control} \label{section: reduction} 
In this section, we state our symbolic controller synthesis problem,
which considers a discrete system obtained by quantising states and
inputs, and by restricting control signals to piecewise-constant ones.
We show that a symbolic controller for this problem also satisfies the
conditions in Definition~\ref{defn: problem continuous}. 

\subsection{Piecewise-constant Control Signal with Discrete Input} \label{subsection: discrete control input}


For a given bounded convex control input space $U \subseteq \R^m$ and a discretisation parameter 
$\mu \in \Rpos$,
let  
\begin{equation}\label{eq: quantise}
U_\mu = \big\{\begin{bmatrix} u_1 & \cdots & u_m
	\end{bmatrix} ^\intercal \in U \bigm| 
	u_i = 2 \mu l_i ,\ l_i \in \Z,\ i \leq m \big\} 
\end{equation}
 be the quantised input set by an $m$-dimensional hypercube of $2\mu$ edge length.
As $U$ is bounded, $U_\mu$ is finite.

Given $\tau\in \Rpos$ and $\ell= [\ell_\text{min}, \ell_\text{max}]$, let us consider a set
\begin{align*}
\UUU_{\tau, \ell, \mu} &=  \bigcup_{ j\tau \in \ell}  \{ u: [0, j\tau] \to U_\mu \} \mid 
j \in  \Zpos \text{ and}
\\
& 
\forall i \in \{0, \ldots,j-1\}, \forall t \in \big[i \tau , (i+1) \tau\big),
u(t) = u(i \tau) 
\}.
\end{align*}
of piecewise-constant control signals.
Each signal $u \in \UUU_{\tau, \ell, \mu}$ is a concatenation of
constant signals of length $\tau$ and value in the finite input set
$U_\mu$.
We limit the length of each signal $u \in \UUU_{\tau, \ell, \mu}$ to be in the range $\ell= [\ell_\text{min}, \ell_\text{max}]$; 
therefore, $\UUU_{\tau, \ell, \mu}$ is also a finite set.

Hence, let us consider a system
$\Sigma_{\tau, \ell, \mu} =(X, X_\ini,$ $ U_\mu, {\UUU}_{\tau, \ell, \mu}, \xi^\fw, \xi^\bw)$,
which is the system $\Sigma$ restricted to piecewise-constant control signals in ${\UUU}_{\tau, \ell, \mu}$. 

\subsection{Symbolic Model and Symbolic Controller} \label{subsection: symbolic model}
A symbolic model is a state-transition model (see Section \ref{subsection: model}) with a discrete state space and a finite set of signals. 
For given $\eta \in \Rpos$, let
\begin{equation}\label{eq: quantise bigger}
  \Dis{X}{\eta} = \big\{x \in \mathbb{R}^n \bigm|
    x_i = 2 \eta l_i ,\ l_i \in \Z,\ i \leq n \,\wedge\,
    \Ball{x}{\eta}\cap X \neq \varnothing\big\}\rlap{.}
\end{equation}
Then, we define a symbolic model of $\Sigma_{\tau, \ell, \mu}$ as follows.

\begin{defn}\label{defn: symbolic model}
Given a system $\Sigma_{\tau, \ell, \mu} = (X, X_\ini, U_\mu, $ ${\UUU}_{\tau, \ell, \mu}, \xi^\fw, \xi^\bw)$ and 
a state-space quantisation parameter $\eta \in \Rpos$,
a \emph{symbolic model} is a state-transition model 
\begin{equation*}
  {\SSS}_{\eta}(\Sigma_{\tau, \ell, \mu}) =
    (Q = \Dis{X}{\eta}, Q_\ini = \Dis{X_\ini}{\eta},
    \UUU_{\tau, \ell, \mu}, \delta) 
\end{equation*} 
such that
$(q ,u, \tilde{q}) \in \delta$ if $(q ,u, \tilde{q}) \in Q \times \UUU_{\tau, \ell, \mu} \times Q$ and
\switchVersion{shortVersion}{
\begin{enumerate} 
\item $\forall x \in \Ball{q}{\eta}$ and $t \leq \len(u)$, $\xi_{x,u}^\fw(t) \subseteq X$, and
\item $\exists \tilde{x} \in \xi_{q,u}^\rightarrow (\len(u))$,
	$
	 	\lVert \tilde{x} - \tilde{q} \rVert \leq \beta_u^\fw(\eta, \len(u)) + \eta, \text{ and}
	$
\item $\exists x \in \xi_{\tilde{q},u}^\leftarrow (\len(u))$,
	$\lVert x - q \rVert \leq \beta_u^\bw(\eta, \len(u)) + \eta.
	$
\end{enumerate}}
{
\begin{enumerate}
\item for all $x \in \Ball{q}{\eta}$ and $t \leq \len(u)$, $\xi_{x,u}^\fw(t) \subseteq X$,
\item there exists $x' \in \xi_{q,u}^\rightarrow (\len(u))$ such that
	\begin{equation*}
	 	\lVert x' - q' \rVert \leq \beta_u^\fw(\eta, \len(u)) + \eta, \text{ and}
	\end{equation*} 
\item there exists $x \in \xi_{q',u}^\leftarrow (\len(u)),$ such that
	\begin{equation*}\lVert x - q \rVert \leq \beta_u^\bw(\eta, \len(u)) + \eta.
	\end{equation*} 
\end{enumerate}}
\end{defn} 

Notice that $\Dis{X}{\eta}$ may contain some points that are not in $X$,
but they have no outgoing transition in $\delta$, so they will not influence our controller synthesis algorithm.

\switchVersion{shortVersion}{}{
\begin{remark}\label{remark: delta_under_approx}
To compute the first condition of $\delta$, we may under-approximate it by 
$\Ball{q}{\alpha^\fw(\eta, \len(u))} \subseteq X$.
\end{remark}
}

A \emph{symbolic controller} is a function
$S:\FRun({\SSS}_{\eta}(\Sigma_{\tau, \ell, \mu})) \to
\UUU_{\tau, \ell, \mu}$ that is a model controller of
${\SSS}_{\eta}(\Sigma_{\tau, \ell, \mu})$.

\begin{remark}
  We can also use multi-dimensional quantisation parameters $\mu =
  \begin{bmatrix} {\mu}_1 & \cdots & {\mu}_m \end{bmatrix}^\intercal
  \in \Rpos^m$ and $\eta = \begin{bmatrix} {\eta}_1 & \cdots &
  {\eta}_n\end{bmatrix} ^\intercal \in \Rpos^n$.
  In this case, Equation~\eqref{eq: quantise} becomes
  \begin{equation*} 
  U_\mu = \big\{\begin{bmatrix} u_1 & \cdots & u_m
  	\end{bmatrix}^\intercal \in U \mid 
  	u_i = 2 \mu_i l_i ,\ l_i \in \Z,\ i \leq m \big\},
  \end{equation*}
  and similarly for Equation~\eqref{eq: quantise bigger}.
\end{remark}

\subsection{Symbolic Control and Approximate Simulation Relation}\label{subsec: approx_sim}
\switchVersion{shortVersion}{
To study the relationship between ${\SSS}_{\eta}(\Sigma_{\tau, \ell, \mu})$ and $\model(\Sigma_{\tau, \ell, \mu})$,
we introduce \emph{alternating approximate simulation relation} between state-transition models,
which is inspired from alternating approximate bisimulation
\cite{pola2009}.
}
{
In order to study the relationship between the symbolic model ${\SSS}_{\eta}(\Sigma_{\tau, \ell, \mu})$ and the state-transition model $\model(\Sigma_{\tau, \ell, \mu})$ of the system $\Sigma_{\tau, \ell, \mu}$,
we introduce a notion of \emph{alternating approximate simulation relation} between state-transition models.}
\begin{defn} \label{defn: simulate}
  Given a pair of transition models $M_1 = (Q_1, Q_{1,\ini}, \VVV,
  \Delta_1)$ and $M_2 = (Q_2, Q_{2,\ini}, \VVV, \Delta_2)$, a metric
  $d: Q_1 \times Q_2 \to \Rnonneg$, and a precision $\varepsilon \in
  \Rnonneg$, $M_1$ alternating \emph{$\varepsilon$-approximately
  simulates} $M_2$ if the following holds:
  \switchVersion{shortVersion}{
  \begin{enumerate}
    \item $\forall x \in Q_{1,\ini}$, $\exists q \in Q_{2,\ini}$ such that $d(x,q) \leq \varepsilon$, and
    \item    
    $\forall (x, q) \in Q_1 \times Q_2$ such that $d(x,q) \leq \varepsilon$,\\
    $\forall u \in \VVV$ such that $\exists (q,u,q') \in \Delta_2$,
    \begin{enumerate}
    \item $\exists (x,u,x') \in \Delta_1$,
    \item $\forall x' \in Q_1$ such that $(x,u,x') \in \Delta_1$,\\ 
    	$\exists (q,u,\tilde{q}) \in \Delta_2$ such that $d(x',\tilde{q}) \leq \varepsilon$.
    \end{enumerate}
  \end{enumerate}
  }{
  \begin{enumerate}
    \item for all $q \in Q_{1,\ini}$, there exists $p \in Q_{2,\ini}$
      with $d(q,p) \leq \varepsilon$,
    \item for each pair $(q, p) \in Q_1 \times Q_2$ such that $d(q,p)
      \leq \varepsilon$, for all $u \in \VVV$ such that there is $(p,u,p') \in \Delta_2$, 
      there is $(q,u,\tilde{q}) \in \Delta_1$ and for all $q'$ with $(q,u,q') \in \Delta_1$, there is 
      $(p,u,p'') \in \Delta_2$ with $d(q',p'') \leq \varepsilon$.
  \end{enumerate}
  }
\end{defn}

If $M_1$ alternating $\varepsilon$-approximately simulates $M_2$, 
then, 
for any signal defined at a state of $M_2$, we have the same signal at the corresponding state of $M_1$. 
Moreover, any non-deterministic behaviour of $M_1$ is also present in $M_2$.
Thus, we can turn any controller of $M_2$ into one of $M_1$.

\begin{lem}\label{lem: approx_simulation} 
Using the metric $d: X \times Q\to \Rnonneg$ 
 given by $d(x,q) = \lVert x - q \rVert$, 
 $\model(\Sigma_{\tau, \ell, \mu})$ alternating $\eta-$approximately simulates the transition model 
${\SSS}_{\eta}(\Sigma_{\tau, \ell, \mu})$.
\end{lem}

\begin{proof}
  The first condition of Definition \ref{defn: simulate} is obvious.
Condition 2)-a) follows from Assumption~\ref{asm: inc fw and bw complete}. 
For 2)-b), we consider 
$\lVert x - q \rVert \leq \eta$ and $(q,u,q') \in\delta$, and assume that $(x,u,x') \in \Delta$.
Since $Q = \Dis{X}{\eta}$, there exists $\tilde{q} \in Q$ such that $\lVert x' - \tilde{q} \rVert \leq \eta$.
We will show that $(q, u, \tilde{q}) \in \delta$.
Condition 1) of Definition~\ref{defn: symbolic model} holds by the fact that $(q,u,q') \in \delta$. 
Then, by Assumption~\ref{asm: inc fw and bw complete}, there exists $(q, u, \tilde{x}) \in \Delta$.
By Assumption~\ref{asm: inc fw and bw complete} and the triangular inequality,
\begin{align*}
\lVert \tilde{x} - \tilde{q} \rVert &\leq \rVert \tilde{x}- x' \lVert + \rVert x' - \tilde{q} \lVert \leq \beta_u^\fw(\lVert x - q \rVert, \len(u)) + \eta, 
\end{align*}
which proves condition 2) of Definition~\ref{defn: symbolic model}.
Condition 3) of Definition~\ref{defn: symbolic model} is shown in the same way.
Consequently,
$(q,u,\tilde{q}) \in \delta$ and therefore condition 2)-b) holds. 
\end{proof}

\subsection{Symbolic Controller Synthesis Problem} \label{subsection: problem symbolic}

\begin{figure}[t]
      \centering
      \includegraphics[scale=0.25]{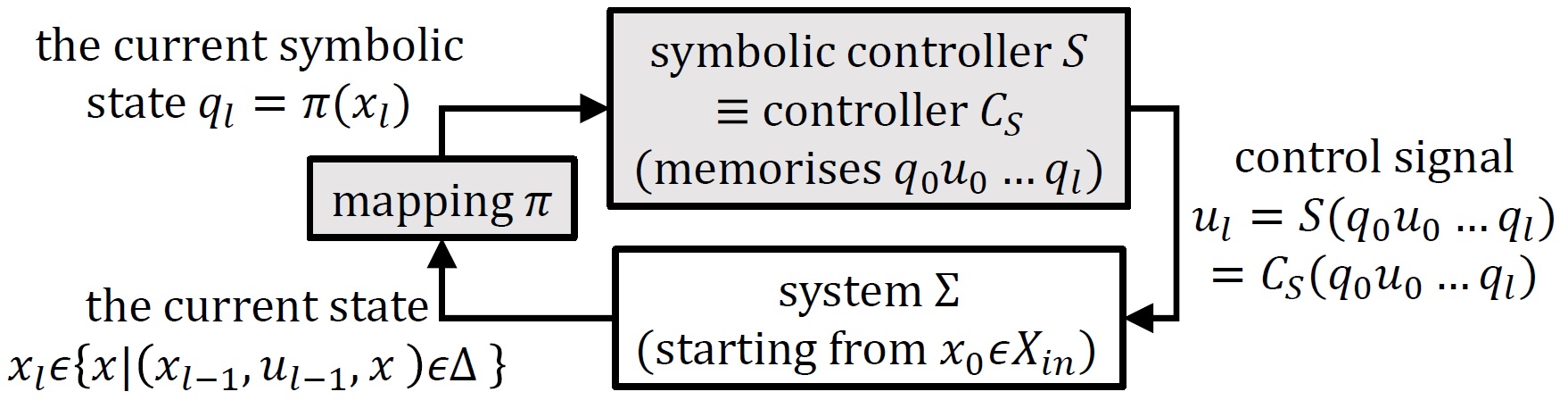}
      \switchVersion{shortVersion}{\squeezeupSmall}{}
      \caption{Overview of the symbolic self-triggered control process. 
      }      \label{fig: symbolic control process}
      \switchVersion{shortVersion}{\squeezeupMid}{}
\end{figure}

In this section, we reduce the controller synthesis problem for the
system $\Sigma_{\tau, \ell, \mu}$ to the synthesis of 
a symbolic controller for ${\SSS}_{\eta}(\Sigma_{\tau, \ell, \mu})$. 
The overview of the symbolic control process is depicted in Fig.~\ref{fig: symbolic control process}.

By Lemma~\ref{lem: approx_simulation}, we can turn any symbolic controller
$S:\FRun({\SSS}_{\eta}(\Sigma_{\tau, \ell, \mu})) \to \UUU_{\tau, \ell, \mu}$
into a controller $\C_S: \FRun(\model(\Sigma_{\tau, \ell, \mu})) \to \UUU_{\tau, \ell, \mu}$. 
More precisely,
let $\pi: X \to Q$ be a mapping such that $x \in \Ball{\pi(x)}{\eta}$.
Then,
for each run $r = x_0 u_1 x_1 \ldots u_l x_l \in  \FRun(\model(\Sigma_{\tau, \ell, \mu}))$,
if $\pi(r) = \pi(x_0) u_1 \pi(x_1) \ldots  u_l \pi(x_l)$, then we
assign $\C_S(r) = S ( \pi(r) )$. 
By Lemma~\ref{lem: approx_simulation}, $\pi(r)$ is a run, and so
$\C_S(r)$ is well defined for any $r \in \FRun(\model(\Sigma_{\tau, \ell, \mu}))$.

\begin{defn} \label{defn: problem symbolic}  
Given a system $\Sigma = (X, X_\ini, U, {\UUU}, \xi^\fw, \xi^\bw)$, a set $AP$ of atomic propositions, a function 
$P: X \to  \powerset(AP)$,
a $\logic$ path formula $\Phi$, 
 a threshold $\nu \in \Rpos$, and quantisation parameters $\tau, \ell, \mu, \eta$, 
the \emph{symbolic controller synthesis problem} consists in synthesising a symbolic controller
$S:\FRun({\SSS}_{\eta}(\Sigma_{\tau, \ell, \mu})) \to \UUU_{\tau, \ell, \mu}$ 
such that 
\begin{itemize}
  \item all finite runs in $\FRun(\control{C_S}{\model(\Sigma)})$ can be
    extended to an infinite run in $\IRun(\model(\Sigma))$,
\item $\traj \vDash \Phi$ for any $\traj \in \Traj(\control{\C_S}{\Sigma})$, and
\item 
$\displaystyle
\lim_{h \to \infty} \frac{1}{h} \sum_{i = 1}^ h \len(u_i) > \nu
$ for any $x_0 u_0 \ldots \in \IRun(\control{\C_S}{\model(\Sigma)})$.
\end{itemize}
or determine that such 
$S$ does not exist.
\end{defn}

By Lemma~\ref{lem: approx_simulation}, we have the following theorem.
\begin{thm}\label{thm: problem_reduction}
If a symbolic controller $S:\FRun({\SSS}_{\eta}(\Sigma_{\tau, \ell, \mu}))$ $\to \UUU_{\tau, \ell, \mu}$  solves the 
problem of Definition~\ref{defn: problem symbolic},
then the controller $\C_S: \FRun(\model(\Sigma_{\tau, \ell, \mu})) \to \UUU_{\tau, \ell, \mu}$ 
solves the controller synthesis problem of Definition~\ref{defn: problem continuous}.
\end{thm}

\switchVersion{shortVersion}{}{
\begin{remark}
The converse of Theorem \ref{thm: problem_reduction} is not necessarily true, and there are models
$\model(\Sigma_{\tau, \ell, \mu})$ that have controllers $C$ that
solve the controller synthesis problem of
Definition~\ref{defn: problem continuous}, but no symbolic controller
solves the problem of Definition~\ref{defn: problem symbolic}.
\end{remark}
}

\section{Control Algorithm}\label{sec: algo}
The overview of the proposed control algorithm is presented in Fig.~\ref{fig: flowchart}.
From the given system $\Sigma$, the signal-length interval $\ell$, and initial quantisation parameters $\eta=\eta_0, \mu = \mu_0, \tau= \tau_0$, we construct the symbolic model ${\SSS}_{\eta}(\Sigma_{\tau, \ell, \mu})$.
Then, we transform the symbolic control problem to a threshold problem of a mean-payoff parity game. 
If there exists a wining strategy of the controller for the game, 
the algorithm translates the strategy to a symbolic controller and terminates.
Otherwise, the algorithm refines the quantisation parameters (e.g., setting $\eta= \frac{\eta}{2}$, or $\mu= \frac{\mu}{2}$, or $\tau= \frac{\tau}{2}$)
and repeats the process.
The algorithm terminates without solving the problem when the parameters $\eta, \mu, \tau$ are smaller than some given thresholds.
\begin{figure}[t]
      \centering
      \includegraphics[scale=0.27]{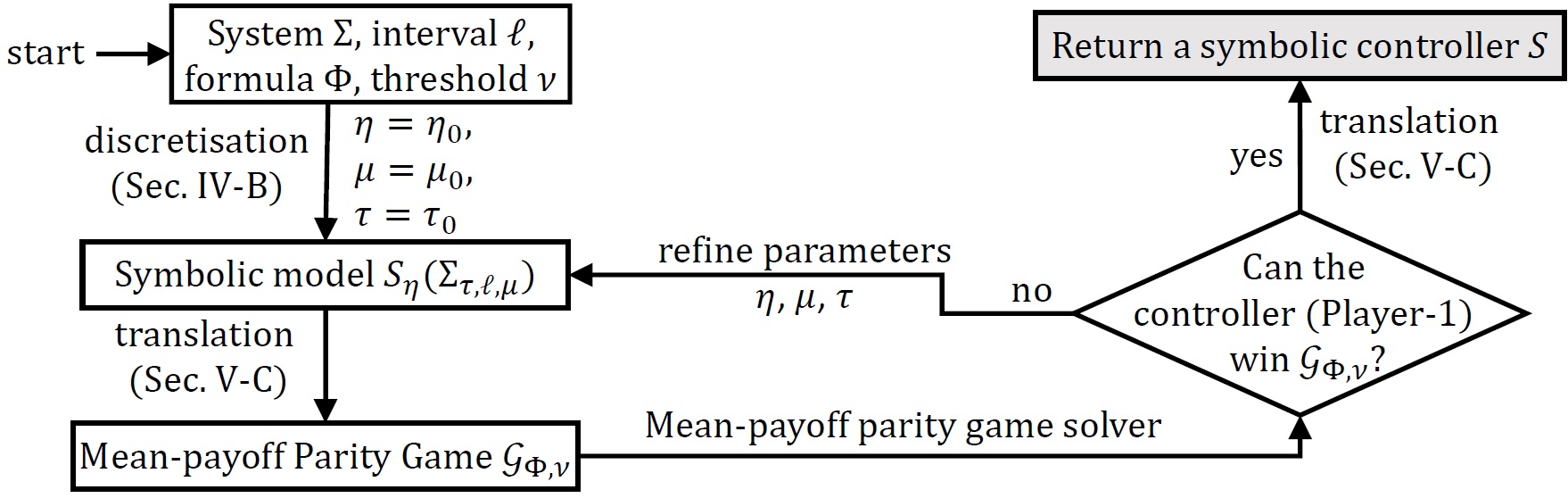}
      \switchVersion{shortVersion}{\squeezeupSmall}{}
      \caption{Overview of the proposed control algorithm.}
      \label{fig: flowchart}
      \switchVersion{shortVersion}{\squeezeupMid}{}
\end{figure}

\subsection{Mean-payoff Parity Game}
Let us first start by recalling known results about \emph{\mppgs{} (MPPGs)}, 
which we use for solving the symbolic controller synthesis problem.
We invite an interested reader to see~\cite{chaterjee05} for more details about MPPGs. 
\begin{defn}
  A \emph{\mppg{}} is a tuple {$\GGG = \big(G=(V = V_1 \sqcup V_2, E, s
      \colon E \to V, t \colon E \to V), \lambda, c, \nu \big)$}, where 
  \begin{itemize}
    \item $G=(V = V_1 \sqcup V_2, E, s
      \colon E \to V, t \colon E \to V)$ is a directed graph.
      	$V$ is partitioned into two disjoint sets $V_1$ and $V_2$ 
      	 of vertices for \emph{Player-1} and \emph{Player-2}, respectively.
      	 $E$ is its set of edges. Functions $s$ and $t$ map edges to their sources and targets.
    \item $\lambda \colon E \to \N$ maps each edge to its
      \emph{payoff}.
    \item $c \colon V \to \N$ maps each vertex to its \emph{colour}.
    \item $\nu \in \N$ is a given \emph{mean-payoff threshold}.
  \end{itemize}
\end{defn}
  A \emph{play} on $\GGG$ is an infinite sequence 
  $\omega = v_0 e_0 v_1 e_1 \ldots \in (VE)^\omega$
  such that, for all $i \geq 0$, $s(e_i)
  = v_i$ and $t(e_i) = v_{i+1}$.
  A \emph{finite play} is a finite sequence in $V(EV)^*$ defined in the same way.
  Let $\FPlay$ be the set of all finite plays, and
  $\FPlayi$ and $\FPlayii$ be the
  sets of finite plays ending with a vertex in $V_1$ and $V_2$, respectively.
Both players play the game by selecting strategies.
A strategy of Player-$i$ is a partial function
  $\sigma_i \colon \FPlay_i \rightharpoonup E$ such that $s(\sigma_i(v_0 e_0 \ldots
   v_n)) = v_n$, i.e., $\sigma_i$ chooses an edge whose
  source is the ending vertex of the play if such an edge exists,
  and does not choose any edge otherwise.
A play $\omega = v_0 e_0 v_1 e_1 \ldots$ is consistent with $\sigma_i$ 
if $e_j = \sigma_i(v_0 e_0 \ldots v_j)$ for all $v_j \in V_i$.
For an initial vertex $v$ and a pair of strategies $\sigma_1$ and $\sigma_2$ of both players,
there exists a unique play, denoted by $\mathit{play}(v, \sigma_1,\sigma_2)$, consistent with both  $\sigma_1$ and $\sigma_2$. This play may be finite if a player cannot choose an 
edge.

For an infinite play $\omega = v_0 e_0 v_1 e_1 \ldots$,
we denote the maximal colour that appears infinitely often in the 
sequence $c(v_0) c(v_1) \ldots$ by $\Inf(\omega)$.
Then, the \emph{mean-payoff value} of the play $\omega$ is  
$
\MP(\omega) = \lim_{n \to \infty} \frac{1}{n} \sum_{i = 1}^n \lambda(e_i)
$.
A vertex $v \in V$ is \emph{winning} for Player-1 
if there exists a strategy
$\sigma_1$ of Player-1 such that,
for any strategy $\sigma_2$ of Player-2, 
$\mathit{play}(v, \sigma_1,\sigma_2)$ is infinite,
$\Inf(\mathit{play}(v, \sigma_1,\sigma_2))$ is even and 
\switchVersion{shortVersion}{$MP(\mathit{play}(v, \sigma_1,\sigma_2))>\nu$.}{$MP(\mathit{play}(v, \sigma_1,\sigma_2))$ is greater than the mean-payoff threshold $\nu$.}
Such a strategy $\sigma_1$ is called a \emph{winning strategy} for Player-1 from the vertex $v$.
Then, the \emph{threshold problem} \cite{Daviaud2018pseudo} is 
to compute the set of winning vertices of Player-1 for a given MPPG.


In \cite{Daviaud2018pseudo}, the authors propose a pseudo-quasi-polynomial algorithm that solves the threshold
problem and computes a winning strategy for Player-1 from each winning state. 
In Section \ref{subsec: problem trans}, 
we reduce the symbolic control problem to the synthesis of a winning strategy on an MPPG,
which can be solved using this algorithm.

\subsection{Atomic Propositions along Symbolic Transitions}

We introduce functions $\pmay,\pmust:\delta \to \powerset(
\powerset(AP) \times \powerset(AP))$
to under-approximate the set of atomic propositions
that hold along trajectories.
This is needed because the information about which states are visited
along a trajectory is lost in the discrete model.
These functions help recover part of this information, and will be
crucial in the problem translation in
Section~\ref{subsec: problem trans}.

For each transition $(q,u,q') \in \delta$ in the symbolic model ${\SSS}_{\eta}(\Sigma_{\tau, \ell, \mu})$, we require that
for $\Finv \in \ens{\forall, \exists}$
{\switchVersion{shortVersion}{\tinydisplayskip
\begin{align} 
\rho_\Finv(q,u,q') \subseteq &\big\{ (P^+,P^-) \in \powerset(AP) \times
    \powerset(AP) \bigm\vert 
   \forall x \in \Ball{q}{\eta}, \nonumber\\
&   \forall x' \in \Ball{q'}{\eta},
   \forall \sigma \in \FTraj(\Sigma_{\tau, \ell, \mu}, x,
    (x u x')),  \label{eq pmay}\\
&  \Finv t \leq \len(u),
  P^+ \subseteq P(\sigma(t)) \wedge P^- \cap
  P(\sigma(t)) = \emptyset \big\}\text{.}\nonumber
\end{align}
}{
\begin{align}
  \rho_\Finv(q,u,q') \subseteq &\big\{ (P^+,P^-) \in \powerset(AP) \times
    \powerset(AP) \bigm\vert \nonumber\\
  & \forall x \in \Ball{q}{\eta}, \forall x' \in \Ball{q'}{\eta},
    \nonumber\\
  & \forall \sigma \in \FTraj({\SSS}_{\eta}(\Sigma_{\tau, \ell, \mu}), x,
    (q u q')), \label{eq pmay}\\
  & \Finv t \leq \len(u), \nonumber 
  P^+ \subseteq P(\sigma(t)) \wedge P^- \cap
  P(\sigma(t)) = \emptyset \big\}\text{.} \nonumber
\end{align}}}
The intuition is as follows. If $(P^+,P^-) \in \rho_\forall(q,u,q')$ (\emph{resp.} $\rho_\exists(q,u,q')$),
then, at all time (\emph{resp.} at some time) along the transition $(q,u,q')$, 
all $p \in P^+$ hold and no $p \in P^-$ holds.
Then, we can define $\rho_\Finv(q,u,q') \vDash \phi$ inductively on the state formula $\phi$ in a sound way.

For the implementation, 
we use functions $B^+, B^-: X \times \R \to \powerset(AP)$  such that, 
for any state $x \in X$ and any radius $r \in \Rpos$, 
 $B^+(x,r) = \{ p \in AP \mid  \forall x' \in \Ball{x}{r}, x' \vDash p\}$
 and
 $B^-(x,r) = \{ p \in AP \mid  \forall x' \in \Ball{x}{r}, x' \nvDash p\}$ 
are the sets of atomic propositions 
that are satisfied and not satisfied, respectively, at all states in the ball $\Ball{x}{r}$.
\begin{asm}\label{asm: P}
For any state $x \in X$ and any radius $r \in \Rpos$, the sets $B^+(x,r)$ and $B^-(x,r)$ can be computed.
\end{asm} 

Then, 
we may use the following functions $\pmay$ and $\pmust$.
{\switchVersion{shortVersion}{\tinydisplayskip}{}
\begin{align*}
\pmay(q,u,q') = &\big(B^+(q,\eta) \cup B^+(q',\eta), B^-(q,\eta) \cup B^-(q',\eta)\big).\\
\pmust(q,u,q') = &\big(B^+(q,r) \cap B^+(q',r), B^-(q,r) \cap B^-(q',r)\big),\\
					&\text{ where } r =\beta^\fw_u(q, \len(u))+ \alpha^\fw_u(q,\len(u)). 
\end{align*}}
By Assumptions \ref{asm: inc fw and bw complete} and \ref{asm: Lipchitz}, $\pmay$ and $\pmust$
satisfy Equation~\eqref{eq pmay}.

\subsection{Problem Translation to Mean-payoff Game}\label{subsec: problem trans}
In this section, we present a translation from the symbolic model 
${\SSS}_\eta(\Sigma_{\tau,\ell,\mu}) = (Q,Q_\ini, \UUU_{\tau,\ell,\mu}, \delta)$, a path formula $\Phi$, and a threshold $\nu \in \Rnonneg$ to an MPPG $\GGG_{\Phi,\nu} = (G_\Phi, \lambda_\Phi, c_\Phi, \nu)$.
The MPPG is played between the controller (as Player-1) and the non-determinism of the system (as Player-2).
The parity constraint will force the controller to induce trajectories
that satisfy $\Phi$, while the mean-payoff constraint will ensure that
the average length of the chosen input signals is above the threshold.
The translation is roughly as illustrated in
Fig.~\eqref{fig:dis_to_game}: Player-1 can move from state $q$ to
state $(q,u)$ (corresponding to choosing input signal $u$), then
Player-2 can choose to go to any $q_i$ reachable from $q$ following
$u$ (corresponding to a non-deterministic environmental behaviour).
The costs on the edges are such that the mean payoff is equal to the
average signal length.
Finally, the nodes' colours are defined inductively on $\Phi$.
\begin{figure}
  \begin{center}
  \resizebox{.3\textwidth}{!}{
    \begin{tikzpicture}
      \node[draw] (l) at (0,0) {
        \begin{tikzpicture}
          \node (ll) at (0,0) {$q$};
          \node (r1) at ($(ll.east)+(1.3,0.6)$) {$q_1$};
          \node (rn) at ($(ll.east)+(1.3,-0.6)$) {$q_n$};
          \node () at ($(r1.center)!0.5!(rn.center)$) {$\vdots$};
          \path[->] (ll) edge node[above] {\small$u$} (r1)
                         edge node[below] {\small$u$} (rn);
        \end{tikzpicture}
      };
      \node[draw,anchor=west] (r) at ($(l.east)+(1.3,0)$) {
        \begin{tikzpicture}
          \node[draw] (ll) at (0,0) {$q$};
          \node[draw,circle] (m) at ($(ll.east)+(1.3,0)$) {$q,u$};
          \node[draw] (r1) at ($(m.east)+(1.3,0.6)$) {$q_1$};
          \node[draw] (rn) at ($(m.east)+(1.3,-0.6)$) {$q_n$};
          \node () at ($(r1.center)!0.5!(rn.center)$) {$\vdots$};
          \path[->] (ll) edge node[above] {\small$\len(u)$} (m)
                    (m) edge node[above] {\small$\len(u)\ \ $} (r1)
                        edge node[below] {\small$\len(u)\ \ $} (rn);
        \end{tikzpicture}
      };
      \node () at ($(l.east)!0.5!(r.west)$) {$\leadsto$};
    \end{tikzpicture}
    }
  \end{center}
  \switchVersion{shortVersion}{\squeezeupSmall}{}
  \caption{Translation to a mean-payoff parity game}
  \label{fig:dis_to_game}
  \switchVersion{shortVersion}{\squeezeupMid}{}
\end{figure}
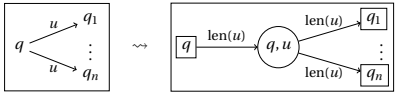

More precisely, we first define a
graph
$G_\Sigma = (V_\Sigma = V_1 \sqcup V_2, E_\Sigma = E_{1\to 2} \sqcup E_{2 \to 1}, s, t)$
and a function $\lambda: E \to \N$ as follows, which corresponds to
what is shown in Fig.~\eqref{fig:dis_to_game}:
\begin{itemize}
  \item $V_1 = Q$ and $V_2 = Q \times \UUU_{\tau,\ell,\mu}$,
  \item $E_{1 \to 2} = Q \times \UUU_{\tau,\ell,\mu}$ and $E_{2 \to 1}
    = \big\{ ((q,u),q') | (q ,u, q') \in \delta \big\}$,
  \item $\forall e = (q,u) \in E_{1 \to 2}$, $q \xto{e} (q,u)$ and
    $\lambda(e) = \len(u)$,
  \item $\forall e = ((q,u),q') \in E_{2 \to 1}$, $(q,u) \xto{e}
    q'$ and $\lambda(e) = \len(u)$.
\end{itemize}
This will form the base of our game $\GGG_{\Phi,\nu}$, which will
roughly consist of multiple copies of $G_\Sigma$, labelled with
different colours, built inductively from $\Phi$.
Technically, we define $\GGG_{\Phi,\nu} = ((Z_\Phi \times V_\Sigma,
Z_\Phi \times E_\Sigma, s_\Phi, t_\Phi), \tilde\lambda, c_\Phi, \nu)$,
where $\tilde\lambda(z,e) = \lambda(e)$, $(z,q) \xto{(z,(q,(q,u)))}
(z,(q,u))$, and $(z,(q,u)) \xto{(z,((q,u),q'))} (z',q')$ for some $z'
\in Z_\Phi$.
The number of copies $Z_\Phi$, the colour function $c_\Phi$, and the target
copy $z'$ are defined inductively.

    For the base cases $\Phi = \Diamond\phi, \Box\phi, \Diamond\Box\phi,
    \Box\Diamond\phi$, we only present $\Box\Diamond \phi$, as the
    other cases are similar.
    We need two copies $Z_\Phi = \{1,2\}$ of $G_\Phi$, labelled with
    colours $c_\Phi(z,v) = z$.
    Finally, $z'$ is $2$ if $\pmay(q,u,q') \vDash \phi$, and $1$
    otherwise.
    The intuition is that we jump to a state in copy $2$
    if we can ensure that there is a state satisfying
    $\phi$ on the trajectory leading to that state, and jump to a
    state in copy $1$ otherwise.
    From there, it is clear that if we can find a strategy that visits
    states in copy $2$ infinitely often (the winning condition for
    parity), then we can force the discrete model to output
    trajectories that satisfy $\Phi$.

    For $\Phi \vee \Psi, \Phi \wedge \Psi$, we
    synchronise parity automata by remembering, for
    each colour (say in $\GGG_{\Psi,\nu}$), the max colour (of
    $\GGG_{\Phi,\nu}$) seen during the execution since a larger colour
    has been seen.
    This allows us to compute the desired $c_{\Phi \vee \Psi}$ and
    $c_{\Phi \wedge \Psi}$. 
    \switchVersion{shortVersion}{}{
    Note that the size of $X_{\Phi \vee\Psi}$ and $X_{\Phi \wedge \Psi}$
    only depend on $X_\Phi$, $X_\Psi$, and the number of colours in
    each MPPG.}

\begin{thm}
  From a winning strategy $\sigma$ for Player-1 in $\GGG_{\Phi,\nu}$,
  one can effectively compute a symbolic controller $C_\sigma$ for
  $\SSS_\eta(\Sigma_{\tau,\ell,\mu})$ that solves the symbolic 
  controller synthesis problem of Definition~\ref{defn: problem symbolic}.
\end{thm}
\switchVersion{shortVersion}{
 \begin{proof} (sketch)
 $C_\sigma$ copies $\sigma$'s choice of input signals.
  The parity condition ensures that the
  controlled system satisfies $\Phi$, while the mean-payoff condition
  ensures that the average signal length is greater than the
  threshold.
\end{proof}
}{
\begin{proof} (sketch)
$C_\sigma$ copies $\sigma$'s choice of input signal.
  The parity condition ensures that any trajectory in the
  controlled system satisfies $\Phi$, while the mean-payoff part
  ensures that the average signal length is greater than the
  threshold.
\end{proof}}
\section{Illustrative Example}
\label{section: illus example}
We consider the following non-deterministic nonholonomic robot system, which is modified version of \cite{SEMGL2019}. 
{\switchVersion{shortVersion}{\tinydisplayskip}{}
\begin{align*} 
\dot{x}(t) &= v (1 + \lambda(t)) \cos(\theta(t)) &\\
\dot{y}(t) &= v (1 + \lambda(t)) \sin(\theta(t)) 
& \dot{\theta}(t) = \omega(t)\rlap{,}
\end{align*}}
where $\omega$ is the input signal for the steering angle, $v$ is the speed of the robot, and $\lambda$ is randomly selected from $[- \bar{\lambda}, \bar{\lambda}]$ for a given parameter $\bar{\lambda} \in \Rnonneg$. 
Notice in particular how this simple system verifies no stability
assumption.

\switchVersion{shortVersion}{
$\xi^\fw$ and $\xi^\bw$ may be over-approximations of
how the physical system behaves.
The non-determinism in the system may come from
the physical system or its mathematical modelling (e.g., to account
for floating-point errors).
For the system described above, the
non-determinism comes from the physical system, where the
velocity of the robot is known only up to some error bound.
}{
\begin{remark}
Note that $\xi^\fw$ and $\xi^\bw$ may always be over-approximations of
how the physical system behaves.
Indeed, we want to synthesise a controller that will force
trajectories of the physical system to verify a given specification,
so it does not matter for soundness if the mathematical modelling of
the physical system yields more trajectories (although it constrains
the controller more).
This means that the non-determinism in the system may come from two
sources: actual non-determinism in the physical system, due to
uncertainties about physical quantities; and non-determinism
introduced by the mathematical modelling of the system.
This last type of non-determinism may serve different purposes, such
as making the set of trajectories easier to compute, or taking into
account the numerical errors introduced when computing the solution of
a differential equation.

In the case of the nonholonomic robot system described above, the
non-determinism comes from the physical system itself, where the
velocity of the robot is only known up to some error bound.
\end{remark}}

Recall that we only consider piecewise-constant input signals in $\UUU_{\tau, \ell, \mu}$ (see Section \ref{subsection: discrete control input}).
For each signal $\omega$ of length $m \tau$, let $\omega_0, \omega_1, \ldots, \omega_{m-1}$
be the constants signals of length $\tau$ such that
$\omega_k(t) = \omega(k\tau + t)$ for any $t \leq \tau$ and $k \in \{0, ..., m-1\}$.
Then, we define $\beta^\fw_\omega$ and $\alpha^\fw_\omega$ as follows. 
{\switchVersion{shortVersion}{\tinydisplayskip}{}
\begin{align*}
&\beta^\fw_\omega(d,k\tau+t) =\\
  &\left\{
    \begin{array}{ll}
      d + 2v (1+\bar\lambda) \sin(\frac{d}{2}) (k\tau + t - \displaystyle\sum_{i=0}^{k-1}
        f_{\omega_i}(\tau) - f_{\omega_k}(t))
        & \hspace{-0.1cm}\text{if $d < \pi$} \\
      d + 2v (1+\bar\lambda) (k\tau + t - \displaystyle\sum_{i=0}^{k-1}
        f_{\omega_i}(\tau) - f_{\omega_k}(t))
        & \hspace{-0.5cm}\text{otherwise,}
    \end{array}
  \right.\\
  &\alpha^\fw_\omega(d,k\tau+t) = d + v (1+\bar\lambda) (k\tau + t -
  \sum_{i=0}^{k-1} f_{\omega_i}(\tau) - f_{\omega_k}(t)),\\
  &\quad\text{where }f_\omega(t) = \left\{
  \begin{array}{ll}
    \frac{\lfloor \frac{\omega t} {\pi} \rfloor (\pi - 2)}{ \omega} & \text{if $\omega \neq
      0$} \\
    0 & \text{otherwise.}
  \end{array}
\right.
\end{align*}} 
Functions $\beta^\bw_\omega$ and $\alpha^\bw_\omega$ are defined in the same way.

We implement our control algorithm using
$v = 2.5$, $\bar{\lambda} = 0.05$,
$X = [-6,6]\times[-6,6]\times[0,2\pi]$,
$X_\ini = \{(0, 0, \frac{\pi}{4})\}$, 
$U = [-\frac{\pi}{2}, \frac{\pi}{2}]$, 
$\eta = \begin{bmatrix} 1 & 1 & \frac{\pi}{8} \end{bmatrix} ^ \intercal$,
$\mu = \frac{\pi}{2}$,
$\tau = \ell_\text{min} = 0.5$, and
$\ell_\text{max} = 1$.  
For the control specification, we set the threshold $\nu = 0.75$ and $\Phi = \Box\Diamond \phi$ where 
$\begin{bmatrix} x & y & \theta \end{bmatrix} ^ \intercal \vDash \phi \iff x > 0 \wedge y>0$.
In other words, we require the average signal length to be greater than $0.75$ 
and the robot to always run back to the green region in Fig.~\ref{fig: robot} after leaving the region.

The program was implemented in Python3.7 and run on a standard laptop computer (Intel
i7-7600U 2.80GHz, 12GB memory). 
The symbolic model contains 3384 states and the mean-payoff game contains 26,400 vertices.
To solve the mean-payoff parity game, we combine the algorithm in \cite{Daviaud2018pseudo} with the algorithms for energy parity games \cite{chatterjee12} and 
mean-payoff games \cite{brim10} to make it more tractable.
All processes took 17 minutes in total. 
Fig.~\ref{fig: robot} shows an example of a finite run under the synthesised controller, 
where the Player-2 (the non-determinism of the system) plays the game by selecting the outgoing edges randomly.

\begin{figure}[t]
      \centering
      \includegraphics[scale=0.27]{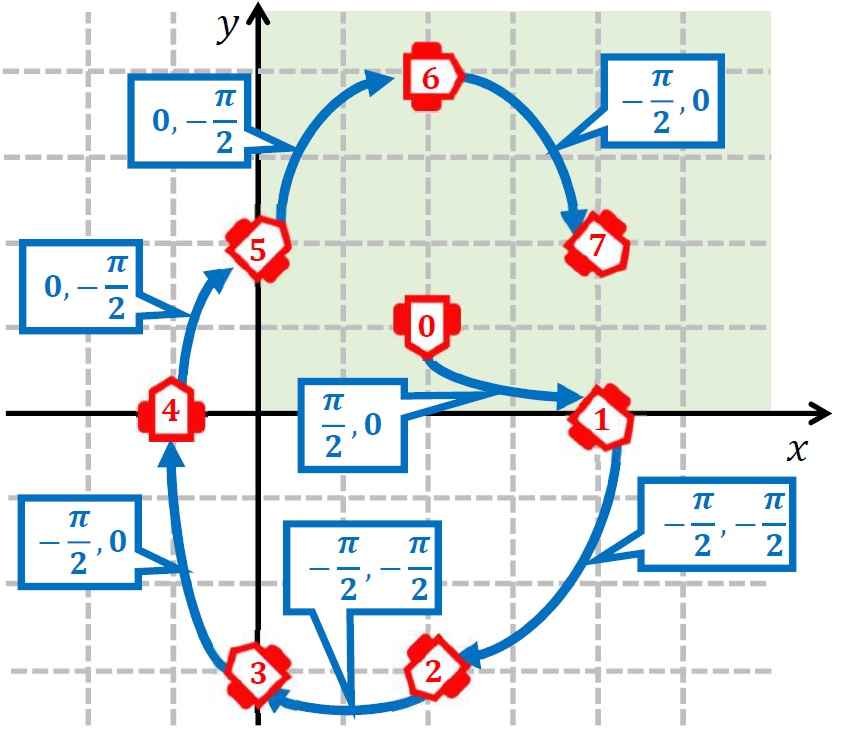}
      \switchVersion{shortVersion}{\squeezeupSmall}{}
      \caption{A finite run under the synthesised controller in 7 time steps. 
      The labels of the arrows show the input signals.
	For example, the arrow from position 0 to 1 represents the signal that assigns
	$\omega = \frac{\pi}{2}$ for $\tau=0.5$ second, and then assigns
	$\omega = 0$ for another $\tau=0.5$ second.  
      }
      \label{fig: robot}
      \switchVersion{shortVersion}{\squeezeupMid}{}
\end{figure}

\section{Conclusion and Future Work}
\label{section: conclusion}

In this paper, we proposed a self-triggered control synthesis
procedure for non-deterministic continuous-time nonlinear systems
without stability assumptions.
The two main ingredients of this procedure are
1) discretising the state and input spaces to obtain a discrete
symbolic model corresponding to the original continuous system
2) reducing the control synthesis problem to the computation of a
winning strategy in a mean-payoff parity game.
We illustrated our method on the example of a nonholonomic robot
navigating in an arena, under a specification requiring it to repeat
some reachability tasks.
As a future work, we would like to expand the size of the considered fragment of LTL. 

\section{ACKNOWLEDGEMENTS}
We thank Prof. Kazumune Hashimoto from Osaka University for his fruitful comments.
\addtolength{\textheight}{-3cm}   




\bibliographystyle{IEEEtran}

\end{document}